\documentclass[11pt]{article}

\usepackage[margin=1in]{geometry}
\usepackage{graphicx}
\usepackage[pdftex,colorlinks=true,linkcolor=blue,citecolor=blue,urlcolor=black]{hyperref}
\usepackage{amsmath, amsthm, amssymb}
\usepackage{subfigure}
\usepackage{comment}
\usepackage{url}
\usepackage{pdflscape}
\usepackage[ruled,lined,linesnumbered]{algorithm2e}

\newcommand{\R}{\mathbb{R}}

\newcommand{\C}{\mathbb{C}}

\newcommand{\E}{\mathbb{E}}

\newcommand{\ket}[1]{| #1 \rangle}

\newcommand{\ip}[2]{\langle #1|#2 \rangle}

\newcommand{\proj}[1]{| #1 \rangle \langle #1 |}
\newcommand{\bracket}[3]{\langle #1|#2|#3 \rangle}
\newcommand{\sm}[1]{\left( \begin{smallmatrix} #1 \end{smallmatrix} \right)}

\DeclareMathOperator{\poly}{poly}

\DeclareMathOperator{\tr}{tr}

\DeclareMathOperator{\Var}{Var}
\DeclareMathOperator{\Inf}{Inf}

\newcommand{\be}{\begin{equation}}
\newcommand{\ee}{\end{equation}}
\newcommand{\bea}{\begin{eqnarray}}
\newcommand{\eea}{\end{eqnarray}}
\newcommand{\bes}{\begin{equation*}}
\newcommand{\ees}{\end{equation*}}
\newcommand{\beas}{\begin{eqnarray*}}
\newcommand{\eeas}{\end{eqnarray*}}

\newtheorem{thm}{Theorem}
\newtheorem*{thm*}{Theorem}
\newtheorem{cor}[thm]{Corollary}

\newtheorem{lem}[thm]{Lemma}
\newtheorem*{lem*}{Lemma}

\begin{document}

\title{Some applications of hypercontractive inequalities in quantum information theory}

\author{Ashley Montanaro\footnote{Centre for Quantum Information and Quantum Foundations, Department of Applied Mathematics and Theoretical Physics, University of Cambridge, UK; {\tt am994@cam.ac.uk}.}}

\maketitle

\begin{abstract}
\noindent Hypercontractive inequalities have become important tools in theoretical computer science and have recently found applications in quantum computation. In this note we discuss how hypercontractive inequalities, in various settings, can be used to obtain (fairly) concise proofs of several results in quantum information theory: a recent lower bound of Lancien and Winter on the bias achievable by local measurements which are 4-designs; spectral concentration bounds for $k$-local Hamiltonians; and a recent result of Pellegrino and Seoane-Sep{\'u}lveda giving general lower bounds on the classical bias obtainable in multiplayer XOR games.
\end{abstract}


\section{Introduction}
\label{sec:intro}

Hypercontractive inequalities, which have been used in theoretical physics for many years, have more recently become increasingly important tools in theoretical computer science. The prototypical example of such an inequality is hypercontractivity of a certain ``noise'' operator on the boolean cube. For a function $f:\{0,1\}^n \rightarrow \R$, and $-1 \le \epsilon \le 1$, define the noise operator $T_\epsilon$ as follows:
\be \label{eq:noiseop} (T_\epsilon f)(x) = \E_{y \sim_\epsilon x}[f(y)], \ee
where the expectation is over bit-strings $y$ obtained from $x$ by flipping each bit of $x$ with probability $(1-\epsilon)/2$. Thus, if $\epsilon=1$, $T_\epsilon f = f$, whereas if $\epsilon=0$, $T_\epsilon f$ is constant. If $f$ has Fourier expansion $f = \sum_{S \subseteq [n]} \hat{f}(S) \chi_S$, where $\chi_S(x) = (-1)^{\sum_{i \in S} x_i}$, then
\[ T_\epsilon f = \sum_{S \subseteq [n]} \epsilon^{|S|} \hat{f}(S) \chi_S. \]
We thus see that $T_\epsilon$ suppresses the higher order Fourier coefficients of $f$. Let $\|f\|_p$ be the normalised $\ell_p$ norm, $\|f\|_p = \left(\frac{1}{2^n} \sum_{x \in \{0,1\}^n} |f(x)|^p \right)^{1/p}$. It is fairly straightforward to show that, for any $p \ge 1$, $\|T_\epsilon f\|_p \le \|f\|_p$, i.e.\ $T_\epsilon$ is a contraction. However, a stronger result also holds.

\begin{thm}[Bonami~\cite{bonami70}, Gross~\cite{gross75a}]
\label{thm:bonamibeckner}
For any $f:\{0,1\}^n \rightarrow \R$, and any $p$ and $q$ such that $1 \le p \le q \le \infty$ and $\epsilon \le \sqrt{\frac{p-1}{q-1}}$,
\[ \|T_\epsilon f\|_q \le \|f\|_p. \]
\end{thm}

Although in general $\|f\|_q \ge \|f\|_p$ for $q \ge p$, we see that for small enough $\epsilon$, $\|T_\epsilon f\|_q \le \|f\|_p$. $T_\epsilon$ is thus said to be {\em hypercontractive}. This precisely accords with the intuition that applying noise to $f$ should make it more smooth; as $p$ increases, $\|f\|_p$ gives more weight to extreme values of $f$, but applying noise smoothes out these extreme values.

Theorem \ref{thm:bonamibeckner} was proven independently by Bonami~\cite{bonami70} and Gross~\cite{gross75a}. In the computer science literature, Theorem \ref{thm:bonamibeckner} is often known as the Bonami-Beckner inequality (Beckner proved generalisations of this result to complex-valued functions and other settings~\cite{beckner75}). The concept of hypercontractivity had its genesis in  important work by Nelson~\cite{nelson66} on quantum field theory (the term ``hypercontractive'' was coined shortly afterwards~\cite{simon72}), and has since found many applications in this and other areas of physics. For detailed reviews from a physics perspective of the history and more recent developments, see~\cite{davies92,gross06}.

Theorem \ref{thm:bonamibeckner} has subsequently also found a number of applications in computer science, the first of which was the celebrated result of Kahn, Kalai and Linial that every boolean function has an influential variable~\cite{kahn88}. Following this, many important results in the theory of boolean functions have also made crucial use of hypercontractivity of $T_\epsilon$ (e.g.~\cite{bourgain02,khot07,mossel10}; see~\cite{dewolf08,odonnell08} for reviews). These results often use Theorem \ref{thm:bonamibeckner} through the following corollary, which relates different norms of low-degree polynomials. For the (simple) proof of this corollary from Theorem \ref{thm:bonamibeckner}, see e.g.~\cite{odonnell07}.

\begin{cor}
\label{cor:hyper}
Let $f:\{\pm1\}^n \rightarrow \R$ be a degree $d$ polynomial. Then, for any $1 \le p \le 2$, $\|f\|_p \ge (p-1)^{d/2} \|f\|_2$, and for any $q \ge 2$, $\|f\|_q \le (q-1)^{d/2} \|f\|_2$.
\end{cor}

Theorem \ref{thm:bonamibeckner} is just one example of a hypercontractive inequality; a more general setting in which to study hypercontractivity is as follows. Let $(S,\mu)$ be a measure space such that $\mu(S) = 1$. For any function $f:S \rightarrow \R$, define the $L^p$ norm of $f$ as $\|f\|_p = \left(\int |f(x)|^p d\mu(x)\right)^{1/p}$, for any $p\ge 1$. Note that these norms are nondecreasing with $p$. Let $L^p(S,\mu)$ be the space of all functions $f:S \rightarrow \R$ such that $\|f\|_p$ is finite. For linear operators $M:L^p(S,\mu) \rightarrow L^q(T,\nu)$, define the operator norm of $M$ as
\[ \|M\| := \sup_{f \neq 0} \frac{\|Mf\|_q}{\|f\|_p}. \]
$M$ is said to be a contraction from $L^p$ to $L^q$ when $\|M\| \le 1$. If $q > p$ and $\|M\| \le 1$, $M$ is said to be hypercontractive. For more on general hypercontractive inequalities, see for example~\cite{davies92,beckner75,beckner92,janson97}.


\subsection{Hypercontractivity in quantum information}

Applications of hypercontractivity are now being found in quantum information theory. In particular, Gavinsky et al.\ used Theorem \ref{thm:bonamibeckner} (via an inequality of Kahn, Kalai and Linial~\cite{kahn88}) to give the first exponential separation between one-way quantum and classical communication complexity of partial boolean functions~\cite{gavinsky07}. Later, Klartag and Regev used hypercontractivity on the $n$-sphere to resolve the long-standing conjecture that one-way quantum communication is exponentially stronger than even two-way classical communication~\cite{klartag11}. Ben-Aroya, Regev and de Wolf generalised Theorem \ref{thm:bonamibeckner} to matrix-valued functions~\cite{benaroya08}, and used their generalised inequality to prove limitations on quantum random access codes. Finally, the original version of Theorem \ref{thm:bonamibeckner} has been used by Buhrman et al.\ to find an upper bound on the classical success probability of a certain non-local game~\cite{buhrman11}, and by Ambainis and de Wolf to give a general lower bound on quantum query complexity~\cite{ambainis12}.

The purpose of the present work is to give several further examples of how hypercontractive inequalities can be used, from a computer science perspective, as tools in quantum information theory. The results themselves are largely not new, and the proofs given here sometimes lead to worse constants than the original proofs. However, the use of hypercontractivity makes the proofs (arguably) more concise and intuitive, and in some cases allows immediate generalisations. Notably, each result we discuss uses hypercontractivity in a different setting (the real $n$-sphere, the space of Hermitian operators on $n$ qubits, and the boolean cube).


\subsubsection{Bias achievable by local measurements}

Given a quantum state which is promised to be either $\rho$ (with probability $p$) or $\sigma$ (with probability $1-p$), it is well known that the optimal measurement for distinguishing $\rho$ and $\sigma$ achieves success probability $\frac{1}{2}\left(1 + \|p \rho - (1-p) \sigma \|_1 \right)$, where $\|M\|_1 = \tr |M|$ is the usual trace norm~\cite{helstrom76,holevo73a}. Setting $\Delta = p \rho - (1-p) \sigma$, we thus obtain that the optimal bias (i.e.\ the difference between the probability of success and failure) over all measurements is just $\|\Delta\|_1$.

However, what if we are not allowed to perform an arbitrary measurement, but are restricted to performing a single fixed quantum measurement, followed by arbitrary classical postprocessing? Given a measurement $M = (M_i)$ (a partition of the identity into positive operators, i.e.\ a POVM), let $\rho^M$, $\sigma^M$ be the probability distributions on measurement outcomes induced by performing $M$ on $\rho$, $\sigma$. The optimal bias one can achieve by performing $M$ is then equal to
\[ \|\Delta\|_M := \|p \rho^M - (1-p) \sigma^M\|_1 = \sum_i |p \tr M_i \rho - (1-p) \tr M_i \sigma| = \sum_i |\tr M_i \Delta|. \]
Generalising work of Ambainis and Emerson~\cite{ambainis07b}, Matthews, Wehner and Winter studied this quantity for measurements which are 4-designs (defined below), and bipartite measurements $M$ comprised of a product of local 4-designs, and gave essentially tight bounds on $\|\Delta\|_M$~\cite{matthews09a}. Very recently, Lancien and Winter have extended these results to general multipartite 4-design measurements~\cite{lancien12}, achieving the following result.

\begin{thm}[Lancien and Winter~\cite{lancien12}]
\label{thm:intromultipartite}
Let $M$ be a $k$-fold tensor product 4-design and set $\Delta = p\rho - (1-p)\sigma$, for quantum states $\rho$, $\sigma \in B((\C^n)^{\otimes k})$. Then there is a universal constant $C > 1$ such that
\[ \|\Delta\|_M \ge C^{-k} \left(\sum_{S \subseteq [k]} \| \tr_S \Delta \|_2^2 \right)^{1/2}, \]
where $\|\Delta\|_2 = \left(\tr |\Delta|^2\right)^{1/2}$ is the Schatten 2-norm.
\end{thm}

Lancien and Winter give two proofs of Theorem \ref{thm:intromultipartite}. While both are arguably elementary (being ultimately based on the use of Cauchy-Schwarz and some combinatorial bounds), the details of each appear somewhat intricate. We give an alternative proof of this result (with a worse constant $C$) based on hypercontractivity on the real $n$-sphere.


\subsubsection{Spectral concentration for $k$-local Hamiltonians}

A Hamiltonian (Hermitian operator) $H$ on the space of $n$ qubits is said to be $k$-local if it can be written as a sum $H = \sum_i H_i$, where each $H_i$ acts non-trivially on at most $k$ qubits. It is well-known that Theorem \ref{thm:bonamibeckner} can be used to prove concentration inequalities for low-degree multivariate polynomials (otherwise known as ``higher order Chernoff bounds'', see e.g.~\cite{dinur07,odonnell08}). We observe that a noncommutative generalisation of Theorem \ref{thm:bonamibeckner} allows an easy proof of the following bound on the spectra of $k$-local operators.

\begin{thm}
\label{thm:introtail}
Let $M$ be a $k$-local Hermitian operator on $n$ qubits with eigenvalues $(\lambda_i)$ such that $\tr M^2 = 2^n$. Then, for any $t \ge (2e)^{k/2}$,
\[ \frac{|\{i: |\lambda_i| \ge t\}|}{2^n} \le \exp(-k t^{2/k} / (2e)). \]
\end{thm}


\subsubsection{General bounds on multiplayer XOR games}

A simple and natural way of exploring the power of quantum correlations is via the framework of XOR games~\cite{clauser69,cleve04}, which have recently been intensively studied (see e.g.~\cite{perezgarcia08,briet09,briet11}). A $k$-player XOR game $G$ is specified as follows. Let $A \in (\{\pm 1\}^n)^k$ be a multidimensional array (tensor). The $j$'th player is given an input $i_j \in \{1,\dots,n\}$ and must reply with an output $x^j_{i_j} \in\{\pm 1\}$. The inputs are picked by a referee according to a known joint probability distribution $\pi$ on $\{1,\dots,n\}^k$. The players win if and only if the product of all their outputs is equal to the corresponding entry of $A$. The maximal bias achievable by deterministic strategies is therefore
\[ \beta(G) := \max_{x^1,\dots,x^k \in \{\pm1\}^n} \left|\sum_{i_1,\dots,i_k=1}^n \pi_{i_1,\dots,i_k} A_{i_1,\dots,i_k} x^1_{i_1} \dots x^k_{i_k} \right|. \]
One can easily see that shared randomness does not help classical players to increase the bias; however, in some cases sharing entanglement can be advantageous~\cite{clauser69}. In the case where $\pi$ is uniform, the problem of calculating or bounding $\beta(G)$ has been studied in several other settings classically, under the title of ``unbalancing lights''~\cite{alon00} or ``Gale-Berlekamp switching games'' (e.g.~\cite{berger97}), and has applications in communication complexity~\cite{chung93,ford05}. Even when $k=2$, $\beta(G)$ is known to be NP-hard to compute~\cite{alon06a,roth08}.

It is an interesting question to determine what the ``hardest'' XOR game is for $k$ classical players, i.e.\ the game for which the bias of the best classical strategy is minimised.  It is known that there exist games for which classical players can achieve a bias of at most $n^{-(k-1)/2}$~\cite{ford05}. On the other hand, an inequality of Bohnenblust and Hille from 1931~\cite{bohnenblust31} implies that any XOR game has maximal bias of at least $2^{-O(k)} n^{-(k-1)/2}$. Based on recent work by Defant, Popa and Schwarting~\cite{defant10}, Pellegrino and Seoane-Sep\'ulveda have significantly improved this result by showing that the exponential dependence on $k$ can be made polynomial.

\begin{thm}[Defant, Popa and Schwarting~\cite{defant10}; Pellegrino and Seoane-Sep\'ulveda~\cite{pellegrino12}]
\label{thm:introxor}
Let $G$ be a $k$-player XOR game with $n$ possible inputs per player. Then there is a universal constant $c>0$ such that $\beta(G) = \Omega(k^{-c} n^{-(k-1)/2})$.
\end{thm}

One can replace a component of their proof with Corollary \ref{cor:hyper}, which leads to a simple and explicit bound on $\beta(G)$ and illustrates the power of using hypercontractivity as a ``black box''. As we discuss below, Theorem \ref{thm:introxor} implies a very special case of a conjecture of Aaronson and Ambainis~\cite{aaronson11} that every bounded low-degree polynomial on the boolean cube has an influential variable.

In the remainder of this paper we elaborate on and prove each of these results in turn.


\section{The bias of local 4-design measurements}

A rank-one POVM $M = (M_i)$ in $n$ dimensions is called a $t$-design~\cite{ambainis07b} if
\be \label{eq:tdesign} \sum_i p_i P_i^{\otimes t} = \int d\psi \proj{\psi}^{\otimes t}, \ee
where $p_i = \frac{1}{n}\tr M_i$ and $P_i = \frac{1}{\tr M_i} M_i$, and the integral is taken according to Haar measure on the complex unit $n$-sphere, normalised such that $\int d\psi \proj{\psi} = I/n$. Observe that the right-hand side of (\ref{eq:tdesign}) is equal to the projector onto the symmetric subspace of $t$ $n$-dimensional systems, normalised by a factor of $\binom{n+t-1}{t}^{-1}$. A $t$-design is automatically an $s$-design for any $1 \le s < t$~\cite{ambainis07b}. $t$-designs can be viewed as discrete approximations to the continuous POVM which puts uniform weight on each measurement vector $\ket{\psi}$. In particular, approximate 4-designs can be used to give an efficient derandomisation of the operation of measurement in a random basis~\cite{ambainis07b}, a primitive which has been used in quantum algorithms~\cite{sen06}.

We will be interested in proving bounds on the bias of $k$-partite measurement operators which are products of local 4-designs. In other words, each measurement operator $M_{i_1,\dots,i_k}$ is a tensor product $M^1_{i_1} \otimes M^2_{i_2} \otimes \dots \otimes M^k_{i_k}$, where each individual measurement $M^i = (M^i_j)$ is a 4-design. We assume for notational simplicity the unnecessary restriction that the local dimensions are the same.

To gain intuition, we begin by considering the unipartite setting $k=1$. In this case we will prove the following theorem, which reproduces a result of Ambainis and Emerson~\cite{ambainis07b} (see also the proof of Matthews, Wehner and Winter~\cite{matthews09a}), with a worse constant.

\begin{thm}
\label{thm:unipartite}
Let $M$ be a 4-design and set $\Delta = p\rho - (1-p)\sigma$, for quantum states $\rho$, $\sigma \in B(\C^n)$ and $0 \le p \le 1$. Then
\[ \|\Delta\|_M \ge \frac{1}{9(1+1/n)^{1/2}} \left( (1-2p)^2 + \tr \Delta^2\right)^{1/2}. \]
\end{thm}

In order to prove Theorem \ref{thm:unipartite} we follow the strategy of \cite{ambainis07b,matthews09a}. To start with, we use the definition of $\|\Delta\|_M$ and H\"older's inequality (in a form popularised by Berger as the ``fourth moment method\footnote{This method seems to have been first used by Littlewood~\cite{littlewood30} in 1930.}''~\cite{berger97}) to obtain
\beas
\|\Delta\|_M &=& \sum_i |\tr M_i \Delta| = n \sum_i p_i |\tr P_i \Delta| \ge n \frac{\left(\sum_i p_i (\tr P_i \Delta)^2\right)^{3/2}}{\left(\sum_i p_i (\tr P_i \Delta)^4\right)^{1/2}}\\
&=& n \frac{\left(\tr \left(\sum_i p_i P_i^{\otimes 2}\right) \Delta^{\otimes 2}\right)^{3/2}}{\left(\tr \left(\sum_i p_i P_i^{\otimes 4}\right) \Delta^{\otimes 4}\right)^{1/2}}.
\eeas
As $M$ is a 4-design (and hence automatically a 2-design), we can replace both the numerator and denominator with the corresponding quantities on the right-hand side of eqn.\ (\ref{eq:tdesign}) to obtain
\[ \|\Delta\|_M \ge n \frac{\left(\tr \left(\int d\psi \proj{\psi}^{\otimes 2}\right) \Delta^{\otimes 2}\right)^{3/2}}{\left(\tr \left(\int d\psi \proj{\psi}^{\otimes 4}\right) \Delta^{\otimes 4}\right)^{1/2}} = n \frac{\left(\int (\tr \Delta \proj{\psi})^2 d\psi \right)^{3/2}}{\left(\int (\tr \Delta \proj{\psi})^4 d\psi\right)^{1/2}}. \]
We now observe that the quantity $\int (\tr \Delta \proj{\psi})^2 d\psi$ can easily be calculated explicitly in terms of the Schatten 2-norm. We have
\beas \tr \left(\int d\psi \proj{\psi}^{\otimes 2}\right) \Delta^{\otimes 2} &=& \tr \left(\frac{I + F}{n(n+1)}\right) \Delta^{\otimes 2} = \frac{1}{n(n+1)}\left( (\tr \Delta)^2 + \tr \Delta^2\right)\\
&=& \frac{1}{n(n+1)}\left( (1-2p)^2 + \tr \Delta^2\right),
\eeas
where $F$ is the swap operator that exchanges two $n$-dimensional systems. Thus, if we can upper bound $\int (\tr \Delta \proj{\psi})^4 d\psi$ in terms of $\int (\tr \Delta \proj{\psi})^2 d\psi$, this will give a lower bound on $\|\Delta\|_M$. Hypercontractivity on the sphere will allow us to do precisely this.


\subsection{Hypercontractivity and spherical harmonics}

For any $n$, let $S^n$ be the real $n$-sphere (i.e.\ $\{x \in \R^{n+1}: \sum_i x_i^2=1\}$), and for $f:S^n \rightarrow \R$ define the $L^p(S^n)$ norms as
\[ \|f\|_{L^p(S^n)} := \left(\int |f(\xi)|^p d\xi \right)^{1/p}, \]
where we integrate with respect to the uniform measure on $S^n$, normalised so that $\int d\xi = 1$. Identify each $n$-dimensional quantum state $\ket{\psi}$ (element of the unit sphere in $\C^n$) with a real vector $\xi \in S^{2n-1}$ by taking real and imaginary parts, and consider the function $f(\xi) = \tr \Delta \proj{\psi}$. It is easy to convince oneself that $f:S^{2n-1} \rightarrow \R$ is a homogeneous degree 2 polynomial in the components of $\xi$, which will allow us to apply hypercontractivity to relate $\int (\tr \Delta \proj{\psi})^4 d\psi$ to $\int (\tr \Delta \proj{\psi})^2 d\psi$.

In order to state and use a hypercontractive inequality on the $n$-sphere, we will need some basic ideas from the theory of spherical harmonics (see e.g.\ the book~\cite{stein71} or the online notes~\cite{gallier09} for background). The restriction of every degree $d$ polynomial $f:\R^{n+1} \rightarrow \R$ to the sphere $S^n$ can be written as
\[ f(x) = \sum_{k=0}^d Y_k(x), \]
where $Y_k:S^n \rightarrow \R$ is called a {\em spherical harmonic}, and is the restriction of a degree $k$ polynomial to the sphere, satisfying $\int Y_j(\xi) Y_k(\xi) d\xi = 0$ for $j \neq k$. Hence
\[ \|f\|_{L^2(S^n)}^2 = \int f(\xi)^2 d\xi = \sum_{k=0}^d \|Y_k\|_{L^2(S^n)}^2, \]
a generalisation of Parseval's equality. The {\em Poisson semigroup} (which can be thought of as a ``noise operator'' for the sphere) is defined by
\[ (P_\epsilon f)(x) = \sum_k \epsilon^k Y_k(x). \]
Crucially, it is known that the Poisson semigroup is indeed hypercontractive.
\begin{thm}[Beckner~\cite{beckner92}]
\label{thm:sphyper}
If $1 \le p \le q \le \infty$ and $\epsilon \le \sqrt{\frac{p-1}{q-1}}$, then
\[ \|P_\epsilon f\|_{L^q(S^n)} \le \|f\|_{L^p(S^n)}. \]
\end{thm}

Theorem \ref{thm:sphyper} will allow us to relate different $L^p$ norms of low-degree polynomials on the sphere, as follows.

\begin{cor}
\label{cor:sphyper}
Let $f:\R^{n+1} \rightarrow \R$ be a degree $d$ polynomial. Then, for any $q \ge 2$,
\[ \|f\|_{L^q(S^n)} \le (q-1)^{d/2} \|f\|_{L^2(S^n)}. \]
\end{cor}

\begin{proof}
The proof is exactly the same as that of Corollary \ref{cor:hyper} (see e.g.~\cite{odonnell07}). Explicitly, we write
\beas
\|f\|_{L^q(S^n)}^2 &=& \left\|\sum_{k=0}^d Y_k \right\|_{L^q(S^n)}^2 = \left\|P_{1/\sqrt{q-1}}\left( \sum_{k=0}^d (q-1)^{k/2} Y_k\right) \right\|_{L^q(S^n)}^2\\
&\le& \left\|\sum_{k=0}^d (q-1)^{k/2} Y_k \right\|_{L^2(S^n)}^2 = \sum_{k=0}^d (q-1)^k \|Y_k\|_{L^2(S^n)}^2 \le (q-1)^d \|f\|_{L^2(S^n)}^2.
\eeas
The first equality is just expanding $f$ in terms of spherical harmonics and the second follows from the definition of $P_\epsilon$. The first inequality is Theorem \ref{thm:sphyper}, the third equality follows from Parseval's theorem, and the last inequality is obvious.
\end{proof}

By Corollary \ref{cor:sphyper}, for any $p \ge 2$,
%
\[ \left( \int (\tr \Delta \proj{\psi})^p d\psi \right)^{1/p} \le (p-1) \left( \int (\tr \Delta \proj{\psi})^2 d\psi \right)^{1/2}. \]
Taking $p=4$ this implies Theorem \ref{thm:unipartite}, via
\beas
\|\Delta\|_M &\ge& n \frac{\left(\int (\tr \Delta \proj{\psi})^2 d\psi \right)^{3/2}}{\left(\int (\tr \Delta \proj{\psi})^4 d\psi\right)^{1/2}} \ge \frac{n}{9} \left(\int (\tr \Delta \proj{\psi})^2 d\psi \right)^{1/2}\\
&=& \frac{1}{9(1+1/n)^{1/2}}\left( (1-2p)^2 + \tr \Delta^2\right)^{1/2}.
\eeas

The works~\cite{ambainis07b,matthews09a} considered the special case where $\Delta$ is traceless and gave a better constant of $1/3$ rather than $1/9$; 
the more recent result of Lancien and Winter~\cite{lancien12} extends the inequality to the case where $\Delta$ is not traceless and achieves a constant $1/\sqrt{18} \approx 1/4.243$. 


\subsection{Hypercontractivity on $k$ copies of the sphere}

We will now apply the approach of the previous section to the multipartite setting, where the power of the hypercontractive approach will become apparent. We will prove the following, reproducing the main result of~\cite{lancien12} with a worse constant.

\begin{thm}
\label{thm:multipartite}
Let $M$ be a $k$-fold tensor product 4-design and set $\Delta = p\rho - (1-p)\sigma$, for quantum states $\rho$, $\sigma \in B((\C^n)^{\otimes k})$. Then
\[ \|\Delta\|_M \ge \frac{1}{(81(1+1/n))^{k/2}}  \left(\sum_{S \subseteq [k]} \| \tr_S \Delta \|_2^2 \right)^{1/2}, \]
where $\|\Delta\|_2$ is the Schatten 2-norm.
\end{thm}

To prove Theorem \ref{thm:multipartite}, we begin by mimicking the start of the proof of Theorem \ref{thm:unipartite} and using the fact that $M$ is a tensor product of local 4-designs to obtain
\be \label{eq:multifrac} \|\Delta\|_M \ge
n^k \frac{\left(\int\dots\int d\psi_1 \dots d\psi_k (\tr \Delta (\proj{\psi_1}\otimes \dots \otimes \proj{\psi_k}) )^2 \right)^{3/2}}{\left(\int\dots\int d\psi_1 \dots d\psi_k (\tr \Delta (\proj{\psi_1}\otimes \dots \otimes \proj{\psi_k}))^4 \right)^{1/2}} = n^k \frac{\|f\|_{L^2((S^{2n-1})^k)}^3}{\|f\|_{L^4((S^{2n-1})^k)}^2},
\ee
where we define the function $f:(S^{2n-1})^k \rightarrow \R$ by
\[ f(\xi_1,\dots,\xi_k) = \tr \Delta(\proj{\psi_1} \otimes \dots \otimes \proj{\psi_k}), \]
where $\ket{\psi_i}$ is the $n$-dimensional complex unit vector whose real and imaginary parts are given by $\xi_i \in S^{2n-1}$ in the obvious way. In order to relate the denominator in (\ref{eq:multifrac}) to the numerator (which can easily be calculated, as we will see below), we require an extension of Theorem \ref{thm:sphyper} and Corollary \ref{cor:hyper} to $(S^{2n-1})^k$. A suitable extension of Theorem \ref{thm:sphyper} immediately follows from multiplicativity of the operator norm of the Poisson operator, which can be proven using standard and generic arguments (see e.g.\ the proof in~\cite{beckner75} or the version in~\cite{dinur05}).

\begin{lem}
Fix $1 \le p \le q$ and consider linear operators $M_1:L^p(S_1,\mu_1) \rightarrow L^q(T_1,\nu_1)$ and $M_2:L^p(S_2,\mu_2) \rightarrow L^q(T_2,\nu_2)$. Then $\|M_1 \otimes M_2\| \le \|M_1\|\|M_2\|$.
\end{lem}

We thus immediately have the following corollary of Theorem~\ref{thm:sphyper}.

\begin{cor}
\label{cor:sphypertp}
Let $f:(S^n)^k \rightarrow \R$. If $1 \le p \le q \le \infty$ and $\epsilon \le \sqrt{\frac{p-1}{q-1}}$, then
\[ \|P_\epsilon^{\otimes k} f\|_{L^q((S^n)^k)} \le \|f\|_{L^p((S^n)^k)}. \]
\end{cor}

This in turn allows us to prove the required generalisation of Corollary \ref{cor:hyper}, which is also almost immediate, but for which we include a proof for completeness.

\begin{cor}
\label{cor:sphconck}
Let $f:(\R^{n+1})^k \rightarrow \R$ be a degree $d$ polynomial in the components of each $x^1,\dots,x^k \in \R^{n+1}$. Then, for any $q \ge 2$,
\[ \|f\|_{L^q((S^n)^k)} \le (q-1)^{dk/2} \|f\|_{L^2((S^n)^k)}. \]
\end{cor}

\begin{proof}
The restriction of $f$ to $(S^n)^k$ can be written as a sum of products of spherical harmonics, i.e.
\[ f(\xi_1,\dots,\xi_k) = \sum_{i_1,\dots,i_k=0}^d Y_{i_1,\dots,i_k}(\xi_1,\dots,\xi_k), \]
where each $Y_{i_1,\dots,i_k}$ is a sum of products of degree $i_1,\dots,i_k$ spherical harmonics on $S^n$, $Y_{i_1,\dots,i_k} = \sum_m Y^m_{i_1} \dots Y^m_{i_k}$, and hence $\int Y_{i_1,\dots,i_k}(\xi_1,\dots,\xi_k) Y_{j_1,\dots,j_k}(\xi_1,\dots,\xi_k)d\xi_1 \dots d\xi_k = 0$ if $i_\ell \neq j_\ell$ for some $\ell$. Then
\beas
\|f\|_{L^q((S^n)^k)}^2 &=& \left\|\sum_{i_1,\dots,i_k=0}^d Y_{i_1,\dots,i_k} \right\|_{L^q((S^n)^k)}^2
= \left\|P_{1/\sqrt{q-1}}^{\otimes k} \left( \sum_{i_1,\dots,i_k=0}^d (q-1)^{\sum_j i_j/2} Y_{i_1,\dots,i_k} \right) \right\|_{L^q((S^n)^k)}^2\\
&\le& \left\|\sum_{i_1,\dots,i_k=0}^d (q-1)^{\sum_j i_j/2} Y_{i_1,\dots,i_k} \right\|_{L^2((S^n)^k)}^2
= \sum_{i_1,\dots,i_k=0}^d (q-1)^{\sum_j i_j} \|Y_{i_1,\dots,i_k}\|_{L^2((S^n)^k)}^2 \\
&\le& (q-1)^{dk} \|f\|_{L^2((S^n)^k)}^2.
\eeas
\end{proof}

We can now finally complete the proof of Theorem \ref{thm:multipartite}. From (\ref{eq:multifrac}) and Corollary \ref{cor:sphconck}, we have
\[ \|\Delta\|_M \ge n^k \frac{\|f\|_{L^2((S^{2n-1})^k)}^3}{\|f\|_{L^4((S^{2n-1})^k)}^2} \ge \left(\frac{n}{9}\right)^k \|f\|_{L^2((S^{2n-1})^k)}. \]
All that remains is to explicitly calculate
\beas
\|f\|_{L^2((S^{2n-1})^k)}^2 &=& \tr \left(\int\dots\int d\psi_1 \dots d\psi_k \proj{\psi_1}^{\otimes 2} \otimes \dots \otimes \proj{\psi_k}^{\otimes 2} \right) \Delta^{\otimes 2} \\
&=& \tr \left(\frac{I + F}{n(n+1)}\right)^{\otimes k} \Delta^{\otimes 2} = \frac{1}{n^k(n+1)^k} \sum_{S \subseteq [k]} \| \tr_S \Delta \|_2^2,
\eeas
which completes the proof of Theorem \ref{thm:multipartite}.

We have therefore reproduced the main result of~\cite{lancien12} (with slightly worse constants). This result is almost optimal, as Lancien and Winter give a simple example ($\Delta = \left(\frac{1}{2}(\proj{0} - \proj{1})\right)^{\otimes k}$) where the bias goes down exponentially with $k$.

It is far from clear that the proof here is ``simpler'' than that of~\cite{lancien12}, which is arguably more intuitive and explicit. One potential advantage of the present approach is that it naturally explains the presence of an exponential prefactor in Theorem \ref{thm:multipartite} in terms of the multiplicativity of the operator norm of the Poisson semigroup. The proof technique here is also somehow more general, as Corollary \ref{cor:sphconck} gives bounds on all $L^q$ norms of $f$, rather than just the $L^4$ norm considered in~\cite{lancien12}, which may be useful in the study of $q$-designs for $q>4$. By contrast, although extending the combinatorial approach of~\cite{ambainis07b,matthews09a,lancien12} to arbitrary $q$ would be possible, it seems likely to require more work. In the special case $k=1$, the approach of~\cite{ambainis07b,matthews09a,lancien12} can indeed be pushed to arbitrary $q$ without too great an effort, and yields better constants than the result given here (Andreas Winter, personal communication).

It is interesting to note that the quantity $\|M\|_{2(k)} := \left(\sum_{S \subseteq [k]} \| \tr_S M \|_2^2\right)^{1/2}$ which appears in Theorem \ref{thm:multipartite} has occurred before in a different setting. If we let $M$ be a quantum state (rather than the difference of two quantum states as considered above) then $\|M\|_{2(k)}^2$ is equal to the probability of $M$ passing a certain natural test for being a product state~\cite{harrow10}, up to a constant depending on $k$ and the local dimensions. The main result of~\cite{harrow10} is an upper bound on $\|M\|_{2(k)}^2$ in terms of the injective tensor norm of $M$, or in other words the maximal overlap of $M$ with a product state. It is an intriguing open question whether the perspective taken here of relating different norms of $M$ could lead to a simpler or stronger proof of the results of~\cite{harrow10}.


\section{Spectral concentration for $k$-local Hamiltonians}

There is a natural extension of the idea of hypercontractivity to a noncommutative setting. Here we will only consider one particular such extension, from functions on the boolean cube $\{0,1\}^n$ to linear operators on the space of $n$ qubits~\cite{qboolean}. Noncommutative hypercontractive inequalities in quantum information have been studied very recently in much greater generality by Kastoryano and Temme, from the perspective of log-Sobolev inequalities~\cite{kastoryano12}.

Let $M \in B((\C^2)^{\otimes n})$ be a Hermitian operator on the space of $n$ qubits. The natural noncommutative analogue of the noise operator defined in eqn.\ (\ref{eq:noiseop}) turns out to be the tensor product of $n$ copies of the qubit depolarising channel $D_\epsilon M = (1-\epsilon) (\tr M) \frac{I}{2} + \epsilon M$. This operator also has a pleasant ``Fourier-side'' description, as follows. Any operator $M \in B((\C^2)^{\otimes n})$ can be expanded in terms of tensor products of Pauli matrices. Identify each string $s \in \{0,1,2,3\}^n$ with the product $\sigma_s := \sigma_{s_1} \otimes \dots \otimes \sigma_{s_n}$ (where $\sigma_0 = \sm{1&0\\0&1}$, $\sigma_1 = \sm{0&1\\1&0}$, $\sigma_2 = \sm{0&-i\\i&0}$, $\sigma_3 = \sm{1&0\\0&-1}$), and write $M = \sum_{s \in \{0,1,2,3\}^n} \widehat{M}(s) \sigma_s$. Then
\[ D_\epsilon^{\otimes n} M = \sum_{s \in \{0,1,2,3\}^n} \epsilon^{|s|} \widehat{M}(s) \sigma_s, \]
where $|s|$ is the number of non-zero components of $s$. Observe that the natural analogue of degree $k$ polynomials on the boolean cube (which have no Fourier coefficients of weight greater than $k$) is operators which have no Pauli coefficients of weight greater than $k$, i.e.\ the class of $k$-local operators.

As one might hope by analogy with the commutative case, $D_\epsilon^{\otimes n}$ does indeed satisfy a hypercontractive inequality~\cite{qboolean}.

\begin{thm}
\label{thm:qhyper}
Let $M$ be a Hermitian operator on $n$ qubits and assume that $1 \le p \le 2 \le q \le \infty$. Then, if $\epsilon \le \sqrt{\frac{p-1}{q-1}}$, $\|D_\epsilon^{\otimes n} M\|_q \le \|M\|_p$.
\end{thm}

In this theorem, and the rest of this section, $\|M\|_p$ is the {\em normalised} Schatten $p$-norm, $\|M\|_p = \left(\frac{1}{2^n} \tr |M|^p \right)^{1/p}$.  Following the completion of this work, King has extended Theorem \ref{thm:qhyper} from the depolarising channel to arbitrary qubit channels that belong to self-adjoint semigroups, and has removed the restriction that $p \le 2 \le q$~\cite{king12}. (Noncommutative hypercontractive inequalities for certain subalgebras of $B((\C^2)^{\otimes n})$ had previously been given by Carlen and Lieb~\cite{carlen93} and Biane~\cite{biane97a}; see~\cite{king12} for a discussion.) The following corollary of Theorem \ref{thm:qhyper} was stated in \cite{qboolean}, by exact analogy with the proof of Corollary \ref{cor:hyper}.

\begin{cor}
\label{cor:qhyper}
Let $M$ be a $k$-local Hermitian operator on $n$ qubits. Then, for any $q \ge 2$, $\|M\|_q \le (q-1)^{k/2} \|M\|_2$. Also, for any $p \le 2$, $\|M\|_p \ge (p-1)^{k/2} \|M\|_2$.
\end{cor}

We now observe that Corollary \ref{cor:qhyper} allows one to easily prove general tail bounds for eigenvalues of $k$-local Hamiltonians, using exactly the same proof as for classical tail bounds on degree $k$ polynomials~\cite{dinur07,odonnell08}.

\begin{thm}
\label{thm:tail}
Let $M$ be a $k$-local Hermitian operator on $n$ qubits with eigenvalues $(\lambda_i)$ such that $\|M\|_2 = 1$. Then, for any $t \ge (2e)^{k/2}$,
\[ \frac{|\{i: |\lambda_i| \ge t\}|}{2^n} \le \exp(-k t^{2/k} / (2e)). \]
\end{thm}

\begin{proof}
For any $r\ge 2$, we have
\[ \frac{|\{i: |\lambda_i| \ge t\}|}{2^n} = \frac{|\{i: |\lambda_i|^r \ge t^r\}|}{2^n} \le \frac{\sum_i |\lambda_i|^r}{2^n t^r} = \frac{\|M\|_r^r}{t^r} \le \frac{((r-1)^{k/2} \|M\|_2)^r}{t^r} \le (r^{k/2}/t)^r, \]
where the first inequality is Markov's inequality, the second is Corollary \ref{cor:qhyper}, and the third follows from the conditions of the theorem. Now, minimising this expression by taking $r=t^{2/k}/e$, we have
\[ \frac{|\{i: |\lambda_i| \ge t\}|}{2^n} \le \exp(-k t^{2/k} / (2e)) \]
as claimed.
\end{proof}

One can, of course, generalise this bound to systems of $d$-dimensional qudits by considering each qudit as $\lceil \log_2 d\rceil$ qubits, at the expense of increasing $k$ to $k\lceil \log_2 d\rceil$. Similar (and/or stronger) concentration bounds to Theorem \ref{thm:tail} have been proven before by other authors in more restricted settings. For example, Hartmann, Mahler and Hess~\cite{hartmann04} have shown that in a setting where subsystems interact with nearest neighbours in a linear chain, the spectrum in fact converges to a normal distribution. Theorem~\ref{thm:tail} gives a less precise result, but makes no assumptions about the interaction geometry and has a much more concise proof.

Physically, Theorem~\ref{thm:tail} says that the probability that a random state has high energy is exponentially small. One consequence is that, for any $k$-local Hamiltonian $H$ on $n$ qubits, there must exist very large subspaces of states on which the time evolution $e^{-iHt}$ is slow. Indeed, one can use Theorem~\ref{thm:tail} to prove bounds on the probability that a state does not significantly change after evolving according to $H$, as in the following example. Let $H$ be a 2-local Hamiltonian on $n$ qubits with eigendecomposition $H = \sum_k \lambda_k \proj{v_k}$, such that $\|H\|_2=1$. Then, for any initial state $\ket{\psi}$ and any $\mu \ge 0$,
\beas |\bracket{\psi}{e^{-iHt}}{\psi}|
&=& \left|\sum_{k,|\lambda_k| \le \mu} e^{-i\lambda_k t} |\ip{v_k}{\psi}|^2 + \sum_{k,|\lambda_k|> \mu} e^{-i\lambda_k t} |\ip{v_k}{\psi}|^2\right|\\
&\ge& \cos(\mu t)\left(1- \sum_{k,|\lambda_k|> \mu} |\ip{v_k}{\psi}|^2\right) - \sum_{k,|\lambda_k|> \mu} |\ip{v_k}{\psi}|^2\\
&\ge& \cos(\mu t) - 2 \sum_{k,|\lambda_k|> \mu} |\ip{v_k}{\psi}|^2.
\eeas
By Theorem \ref{thm:tail}, for any $\mu \ge 2e$ the dimension of the subspace spanned by $\ket{v_k}$ such that $|\lambda_k|> \mu$ is at most $e^{-\mu/e} 2^n$. Now imagine we pick $\ket{\psi}$ at random (i.e.\ according to Haar measure on the complex unit sphere). This implies, via a standard tail bound for projector overlaps~\cite{bennett05}, that for any $\delta>0$
\[ \Pr_{\ket{\psi}}\left[\sum_{k,|\lambda_k|> \mu} |\ip{v_k}{\psi}|^2 \ge (1+\delta) e^{-\mu/e}\right] \le \exp(-e^{-\mu/e} 2^n(\delta-\ln(1+\delta))/(\ln 2) ). \]
Fixing (for example) $\delta=1$, we obtain for any $\mu \ge 2e$ that
\[ \Pr_{\ket{\psi}}\left[ |\bracket{\psi}{e^{-iHt}}{\psi}| \le \cos(\mu t) - 4e^{-\mu/e} \right] \le \exp(-\Omega(e^{-\mu/e} 2^n)). \]
Other spectral inequalities follow from Corollary \ref{cor:qhyper}. For example, the following quantum generalisation of the classical Schwartz-Zippel lemma was observed in~\cite{qboolean}.

\begin{cor}
Let $H$ be a non-zero $k$-local Hermitian operator on $n$ qubits with rank $r$. Then $r \ge 2^{n - (2\log_2 e)k} \approx 2^{n - 2.89k}$.
\end{cor}


\section{General bounds on XOR games}

A homogeneous polynomial $f:(\R^n)^k \rightarrow \R$ is said to be a {\em multilinear form} if it is linear in each input, i.e.\ $f(x^1+y,x^2,\dots,x^k) = f(x^1,\dots,x^k) + f(y,x^2,\dots,x^k)$ for all $x^i\in \R^n$ and all $y\in \R^n$, and similarly for the other positions. Any multilinear form can be written as
\[ f(x^1,\dots,x^k) = \sum_{i_1,\dots,i_k} \hat{f}_{i_1,\dots,i_k} x^1_{i_1} x^2_{i_2} \dots x^k_{i_k} \]
for some multidimensional array (tensor) $\hat{f} \in \R^n \times \R^n \times \dots \times \R^n$. Observe that, if we consider $x^j \in \{\pm1\}^n$,  this expression is precisely the Fourier expansion of $f$ as a function on the boolean cube $\{\pm1\}^{nk}$, justifying the use of the notation $\hat{f}$. The (normalised) $\ell_p$ norms of $f$ as a function on the boolean cube are thus given by
\[ \|f\|_p := \left(\frac{1}{2^{nk}} \sum_{x^1,\dots,x^k \in \{\pm1\}^n} |f(x^1,\dots,x^k)|^p \right)^{1/p}, \]
and in particular
\[ \|f\|_\infty := \max_{x^1,\dots,x^k \in \{\pm1\}^n} |f(x^1,\dots,x^k)|. \]
Any XOR game $G = (\pi,A)$ corresponds to a multilinear form $f$ by taking
\[ f(x^1,\dots,x^k) = \sum_{i_1,\dots,i_k} \pi_{i_1,\dots,i_k} A_{i_1,\dots,i_k} x^1_{i_1} x^2_{i_2} \dots x^k_{i_k}, \]
and the bias $\beta(G)$ is precisely $\|f\|_\infty$. The following result, which is known as the Bohnenblust-Hille inequality, will allow us to find a general lower bound on $\beta(G)$.

\begin{thm}[Bohnenblust-Hille inequality~\cite{bohnenblust31,defant10,pellegrino12}]
\label{thm:forms}
For any multilinear form $f:(\R^n)^k \rightarrow \R$, and any $p \ge 2k/(k+1)$,
\[ \|\hat{f}\|_p := \left( \sum_{i_1,\dots,i_k} |\hat{f}_{i_1,\dots,i_k}|^p \right)^{1/p} \le C_k \|f\|_\infty, \]
where $C_k$ may be taken to be $O(k^{\log_2 e}) \approx O(k^{1.45})$.
\end{thm}

The special case $k=2$ was previously proven by Littlewood~\cite{littlewood30} and is known as Littlewood's $4/3$ inequality. The original proof of Theorem \ref{thm:forms} given by Bohnenblust and Hille had $C_k$ growing exponentially with $k$~\cite{bohnenblust31}. Pellegrino and Seoane-Sep\'ulveda have recently shown that in fact one can take $C_k = \poly(k)$ in this inequality~\cite{pellegrino12} (see also~\cite{diniz12,nunezalarcon12,serranorodriguez12} for further and more precise bounds on $C_k$). The proof in~\cite{pellegrino12} is formally a consequence of prior, and very general, work by Defant, Popa and Schwarting~\cite{defant10}, but makes various careful combinatorial choices to achieve an explicit bound on $C_k$ which is polynomial in $k$; we therefore give credit for Theorem \ref{thm:forms} to both sets of authors. Here we will give a modified proof,  based on hypercontractivity, of this result.

First, we observe that Theorem \ref{thm:forms} has the following corollary for XOR games.

\begin{cor}
Let $G$ be a $k$-player XOR game with $n$ possible inputs per player. Then $\beta(G) = \Omega(k^{-3/2} n^{-(k-1)/2})$.
\end{cor}

\begin{proof}
Apply Theorem \ref{thm:forms} to the multilinear form $f$ corresponding to $G$, taking $p=2k/(k+1)$, and use the inequality $\|\hat{f}\|_p \ge n^{k(1/p-1)}\|\hat{f}\|_1 = n^{k(1/p-1)}$.
\end{proof}

In order to prove Theorem \ref{thm:forms} we will need the following inequality of Defant, Popa and Schwarting~\cite{defant10} (extending an inequality of Blei~\cite{blei01}), which can be proven by the careful application of H\"older's inequality.

\begin{lem}
\label{lem:matineq}
Let $A = (A_{ij})$ be a matrix whose columns are $(\alpha_i)$ and whose rows are $(\beta_j)$. Then, for any $m \ge 1$,
\[ \left(\sum_{i,j} |A_{ij}|^{2m/(m+1)} \right)^{(m+1)/(2m)} \le \left(\sum_i \|\alpha_i\|_2^{2m/(m+2)} \right)^{(m+2)/(4m)} \left(\sum_j \|\beta_j\|_2^{2m/(m+2)} \right)^{(m+2)/(4m)}. \]
\end{lem}

As with the previous two applications of hypercontractivity discussed, the proof of Theorem \ref{thm:forms} that we give here will also use hypercontractivity via a corollary: in this case the perhaps more well-known Corollary \ref{cor:hyper} for functions on the boolean cube. We stress that the proof given here follows the same lines as that of Pellegrino and Seoane-Sep\'ulveda; the use of Corollary \ref{cor:hyper} simply replaces an equivalent step in their proof (and indeed gives worse constants). However, using this corollary seems (to the author) to make the proof somewhat simpler and more transparent, and illustrates how  hypercontractivity can be used as a ``black box''.

\begin{proof}[Proof of Theorem \ref{thm:forms} (Defant, Popa and Schwarting~\cite{defant10}; Pellegrino and Seoane-Sep\'ulveda~\cite{pellegrino12})]
$ $\\ As $\|\hat{f}\|_p$ is nonincreasing with $p$, it suffices to prove the theorem for $p=2k/(k+1)$. The proof will be by induction on $k$, first assuming that $k$ is a power of 2. The base case $k=1$ is trivial, and in this case we have $C_1=1$. So, assuming the theorem holds for $k/2$, we prove it holds for $k$. By Lemma \ref{lem:matineq},
\bea \left( \sum_{i_1,\dots,i_k} |\hat{f}_{i_1,\dots,i_k}|^{2k/(k+1)} \right)^{(k+1)/(2k)} &\le& \left(\sum_{i_1,\dots,i_{k/2}} \|(\hat{f}_{i_1,\dots,i_k})_{i_{k/2+1},\dots,i_k=1}^n\|_2^{2k/(k+2)} \right)^{(k+2)/4k}\nonumber \\ 
&& \label{eq:twoterms} \times \left(\sum_{i_{k/2+1},\dots,i_k} \|(\hat{f}_{i_1,\dots,i_k})_{i_1,\dots,i_{k/2}=1}^n\|_2^{2k/(k+2)} \right)^{(k+2)/4k}
\eea
We estimate the second term (the first follows exactly the same procedure). For each $i_{k/2+1},\dots,i_k \in \{1,\dots,n\}$, define the function $f_{i_{k/2+1},\dots,i_k}:(\R^n)^{k/2} \rightarrow \R$ by
\[ f_{i_{k/2+1},\dots,i_k}(x^1,\dots,x^{k/2}) = \sum_{i_1,\dots,i_{k/2}} \hat{f}_{i_1,\dots,i_k} x^1_{i_1} x^2_{i_2} \dots x^{k/2}_{i_{k/2}}. \]
Also define a ``dual'' function $f'_{x^1,\dots,x^{k/2}}:(\R^n)^{k/2} \rightarrow \R$ by
\[ f'_{x^1,\dots,x^{k/2}}(x^{k/2+1},\dots,x^k) = f(x^1,\dots,x^k), \]
i.e.\ with respect to $\pm1$-valued inputs, $f'_{x^1,\dots,x^{k/2}}$ is just the restriction of $f$ to the subcube produced by fixing $x^1,\dots,x^{k/2}$. It is easy to verify that $f'_{x^1,\dots,x^{k/2}}$ can be written as
\[ f'_{x^1,\dots,x^{k/2}}(x^{k/2+1},\dots,x^k) = \sum_{i_{k/2+1},\dots,i_k=1}^n f_{i_{k/2+1},\dots,i_k}(x^1,\dots,x^{k/2}) x^{k/2+1}_{i_{k/2+1}} \dots x^k_{i_k}; \]
of course $\|f'_{x^1,\dots,x^{k/2}}\|_{\infty} \le \|f\|_{\infty}$. For each tuple $i_{k/2+1},\dots,i_k$ we have by Parseval's equality
\[ \|(\hat{f}_{i_1,\dots,i_k})_{i_1,\dots,i_{k/2}=1}^n\|_2 = \left(\sum_{i_1,\dots,i_{k/2}=1}^n \hat{f}_{i_1,\dots,i_k}^2 \right)^{1/2} = \|f_{i_{k/2+1},\dots,i_k}\|_2. \]
By the hypercontractive estimate of Corollary \ref{cor:hyper},
\[\|f_{i_{k/2+1},\dots,i_k}\|_2^{2k/(k+2)} \le \left(\frac{k+2}{k-2} \right)^{\frac{k^2}{2(k+2)}} \|f_{i_{k/2+1},\dots,i_k}\|_{2k/(k+2)}^{2k/(k+2)}. \]
We now observe that, for any $p \ge 1$,
\beas
\sum_{i_{k/2+1},\dots,i_k} \|f_{i_{k/2+1},\dots,i_k}\|_p^p &=& \E_{x^1,\dots,x^{k/2}}\left[ \sum_{i_{k/2+1},\dots,i_k} |f_{i_{k/2+1},\dots,i_k}(x^1,\dots,x^{k/2})|^p \right]\\
&=& \E_{x^1,\dots,x^{k/2}}\left[ \|\hat{f'}_{x^1,\dots,x^{k/2}}\|^p_p \right].
\eeas
Hence, taking $p = 2k/(k+2) = 2(k/2)/(k/2+1)$, we have by the inductive hypothesis
\beas
\sum_{i_{k/2+1},\dots,i_k} \|(\hat{f}_{i_1,\dots,i_k})_{i_1,\dots,i_{k/2}=1}^n\|_2^{2k/(k+2)} &\le& \E_{x^1,\dots,x^{k/2}}\left[ \|\hat{f'}_{x^1,\dots,x^{k/2}}\|_{2k/(k+2)}^{2k/(k+2)} \right] \\
&\le& \left(\frac{k+2}{k-2} \right)^{\frac{k^2}{2(k+2)}} C_{k/2}^{2k/(k+2)} \|f\|_{\infty}^{2k/(k+2)},
\eeas
so, combining both terms in the inequality (\ref{eq:twoterms}),
\[ \left( \sum_{i_1,\dots,i_k} |\hat{f}_{i_1,\dots,i_k}|^{2k/(k+1)} \right)^{(k+1)/(2k)} \le \left(\frac{k+2}{k-2} \right)^{k/4} C_{k/2} \|f\|_{\infty}. \]
Thus
\[ C_k \le \left(1 + \frac{4}{k-2} \right)^{k/4} C_{k/2}. \]
%
Observing that $\left(1 + 4/(k-2) \right)^{k/4} \le (1+O(1/k))e$, we have $C_k = O(k^{\log_2 e})$ as claimed.

Finally, if $k$ is not a power of 2, we simply increase it to the next smallest power of 2 (redefining $f$ appropriately), which corresponds to at most a constant increase in $C_k$.
\end{proof}

We remark that Theorem \ref{thm:forms} also proves a very special case of a conjecture of Aaronson and Ambainis~\cite{aaronson11} that every bounded low-degree polynomial on the boolean cube has an influential variable. Define the {\em influence} of the $j$'th variable on a function $f:\{\pm1\}^n \rightarrow \R$ as $I_j(f) = \frac{1}{2^{n+2}} \sum_{x \in \{\pm1\}^n} (f(x)-f(x^j))^2$, where $x^j$ is $x$ with the $j$'th variable negated. The influence also has a concise Fourier-side description: $I_j(f) = \sum_{S \ni j} \hat{f}(S)^2$. The conjecture of~\cite{aaronson11} is that for all degree $d$ polynomials $f:\{\pm 1\}^n \rightarrow [-1,1]$, there exists a $j$ such that $I_j(f) \ge \poly(\Var(f)/d)$, where $\Var(f) = \sum_{S \neq \emptyset} \hat{f}(S)^2$ is the $\ell_2$ variance of $f$. If this conjecture were true, it would imply (informally) that all quantum query algorithms could be efficiently simulated by classical query algorithms on most inputs~\cite{aaronson11}.

Using Theorem \ref{thm:forms}, it is easy to prove the Aaronson-Ambainis conjecture in the very particular case where $f$ is a multilinear form such that $\hat{f}_{i_1,\dots,i_k} = \pm \alpha$, for some $\alpha$, as we now show. If $f$ is a multilinear form as above, it depends on $nk$ variables $x^j_\ell$, where $1 \le j \le k$ and $1 \le \ell \le n$. Observe that the influence of variable $(j,\ell)$ on $f$ is
\[ \Inf_{(j,\ell)}(f) = \sum_{i_1,\dots,i_{j-1},i_{j+1},\dots,i_k} \hat{f}_{i_1,\dots,i_{j-1},\ell,i_{j+1},\dots,i_k}^2. \]
We can now state the following corollary of Theorem \ref{thm:forms}.

\begin{cor}
If $f$ is a multilinear form such that $\|f\|_{\infty} \le 1$ and $\hat{f}_{i_1,\dots,i_k} = \pm \alpha$ for some $\alpha$, then $I_{(j,\ell)}(f) = \Omega(\Var(f)^2 / k^3)$ for all $(j,\ell)$.
\end{cor}

\begin{proof}
We simply calculate $\Var(f) = n^k \alpha^2$, $I_{(j,\ell)}(f) = n^{k-1} \alpha^2$, and $\|\hat{f}\|_{2k/(k+1)} = n^{(k+1)/2} \alpha = O(k^{3/2})$
%
%
by Theorem \ref{thm:forms}. Thus $\Var(f)^2 / I_{(j,\ell)}(f) = n^{k+1} \alpha^2 = O(k^3)$ as required.
\end{proof}

Andris Ambainis has subsequently generalised this result to multilinear forms $f$ such that $\hat{f}_{i_1,\dots,i_k} \in \{-\alpha,0,\alpha\}$, and $\hat{f}_{i_1,\dots,i_k} \neq 0$ at $n^{\Omega(k)}$ positions $(i_1,\dots,i_k)$ (personal communication). Generalising further to arbitrary multilinear forms might be an interesting way of making progress on the Aaronson-Ambainis conjecture.


\section{Outlook}

This work has presented several examples of results in quantum information theory which can be proven using hypercontractive inequalities. It is, of course, debatable as to whether the proofs given here are really simpler or more intuitive than previously known proofs, especially if one is not initially familiar with hypercontractivity. The author's feeling is that the proofs given here seem less technical than the original proofs, at the expense of being less explicit; it is hoped that hypercontractivity will continue to develop as a tool in quantum information theory.


\section*{Acknowledgements}

This work was supported by an EPSRC Postdoctoral Research Fellowship. I would like to thank Marius Junge, David P\'erez-Garc\'ia and Andreas Winter for pointing out~\cite{biane97a}, \cite{serranorodriguez12} and \cite{hartmann04} respectively; Aram Harrow, Daniel Pellegrino, Juan Seoane-Sep\'ulveda and Ronald de Wolf for helpful comments on previous versions; and also Andris Ambainis, Toby Cubitt, Will Matthews and Andreas Winter for helpful discussions. Finally, I would like to thank two referees for their constructive comments.



\begin{thebibliography}{10}

\bibitem{aaronson11}
S.~Aaronson and A.~Ambainis.
\newblock The need for structure in quantum speedups.
\newblock In {\em Proceedings of ICS 2011}, pages 338--352, 2011.
\newblock \url{arXiv:0911.0996}.

\bibitem{alon06a}
N.~Alon and A.~Naor.
\newblock Approximating the cut-norm via {G}rothendieck's inequality.
\newblock {\em SIAM J. Comput.}, 35(4):787--803, 2006.

\bibitem{alon00}
N.~Alon and J.~Spencer.
\newblock {\em The probabilistic method}.
\newblock Wiley-Interscience Series in Discrete Mathematics and Optimization.
  Wiley, New York, 2000.

\bibitem{ambainis12}
A.~Ambainis and R.~de~Wolf.
\newblock How low can approximate degree and quantum query complexity be for
  total boolean functions?, 2012.
\newblock \url{arXiv:1206.0717}.

\bibitem{ambainis07b}
A.~Ambainis and J.~Emerson.
\newblock Quantum t-designs: t-wise independence in the quantum world.
\newblock In {\em Proc. 22\textsuperscript{nd} Annual IEEE Conf. Computational
  Complexity}, pages 129--140, 2007.
\newblock \url{quant-ph/0701126}.

\bibitem{beckner75}
W.~Beckner.
\newblock Inequalities in {F}ourier analysis.
\newblock {\em Ann. of Math.}, 102:159--182, 1975.

\bibitem{beckner92}
W.~Beckner.
\newblock {S}obolev inequalities, the {P}oisson semigroup, and analysis on the
  sphere {$S^n$}.
\newblock {\em Proc. Natl. Acad. Sci. USA}, 89:4816--4819, 1992.

\bibitem{benaroya08}
A.~Ben-Aroya, O.~Regev, and R.~de~Wolf.
\newblock A hypercontractive inequality for matrix-valued functions with
  applications to quantum computing and {LDC}s.
\newblock In {\em Proc. 49\textsuperscript{th} Annual Symp. Foundations of
  Computer Science}, pages 477--486, 2008.
\newblock \url{arXiv:0705.3806}.

\bibitem{bennett05}
C.~H. Bennett, P.~Hayden, D.~Leung, P.~Shor, and A.~Winter.
\newblock Remote preparation of quantum states.
\newblock {\em IEEE Trans. Inform. Theory}, 51(1):56--74, 2005.
\newblock \url{quant-ph/0307100}.

\bibitem{berger97}
B.~Berger.
\newblock The fourth moment method.
\newblock {\em SIAM J. Comput.}, 24(6):1188–--1207, 1997.

\bibitem{biane97a}
P.~Biane.
\newblock Free hypercontractivity.
\newblock {\em Comm. Math. Phys.}, 184(2):457--474, 1997.

\bibitem{blei01}
R.~Blei.
\newblock {\em Analysis in integer and fractional dimensions}.
\newblock Cambridge University Press, 2001.

\bibitem{bohnenblust31}
H.~Bohnenblust and E.~Hille.
\newblock On the absolute convergence of {D}irichlet series.
\newblock {\em Annals of Mathematics}, 32(3):600--622, 1931.

\bibitem{bonami70}
A.~Bonami.
\newblock \'{E}tude des coefficients {F}ourier des fonctiones de {Lp(G)}.
\newblock {\em Ann. Inst. Fourier}, 20:335--–402, 1970.

\bibitem{bourgain02}
J.~Bourgain.
\newblock On the distribution of the {F}ourier spectrum of {B}oolean functions.
\newblock {\em Isr. J. Math.}, 131(1):269--276, 2002.

\bibitem{briet09}
J.~Bri{\"e}t, H.~Buhrman, T.~Lee, and T.~Vidick.
\newblock Multiplayer {XOR} games and quantum communication complexity with
  clique-wise entanglement, 2009.
\newblock \url{arXiv:0911.4007}.

\bibitem{briet11}
J.~Bri{\"e}t and T.~Vidick.
\newblock Explicit lower and upper bounds on the entangled value of multiplayer
  {XOR} games, 2011.
\newblock \url{arXiv:1108.5647}.

\bibitem{buhrman11}
H.~Buhrman, O.~Regev, G.~Scarpa, and R.~de~Wolf.
\newblock Near-optimal and explicit {B}ell inequality violations.
\newblock In {\em Proc. 26\textsuperscript{th} Annual IEEE Conf. Computational
  Complexity}, pages 157--166, 2011.
\newblock \url{arXiv:1012.5043}.

\bibitem{carlen93}
E.~Carlen and E.~Lieb.
\newblock Optimal hypercontractivity for {F}ermi fields and related
  non-commutative integration inequalities.
\newblock {\em Comm. Math. Phys.}, 155:27--46, 1993.

\bibitem{chung93}
F.~Chung and P.~Tetali.
\newblock Communication complexity and quasi randomness.
\newblock {\em SIAM J. Discrete Math.}, 6(1):110--125, 1993.

\bibitem{clauser69}
J.~Clauser, M.~Horne, A.~Shimony, and R.~Holt.
\newblock Proposed experiment to test local hidden-variable theories.
\newblock {\em Phys. Rev. Lett.}, 23:880--884, 1969.

\bibitem{cleve04}
R.~Cleve, P.~H{\o }yer, B.~Toner, and J.~Watrous.
\newblock Consequences and limits of nonlocal strategies.
\newblock In {\em Proc. 19\textsuperscript{th} Annual IEEE Conf. Computational
  Complexity}, pages 236--249, 2004.
\newblock \url{quant-ph/0404076}.

\bibitem{davies92}
E.~B. Davies, L.~Gross, and B.~Simon.
\newblock Hypercontractivity: a bibliographic review.
\newblock {\em Ideas and methods in quantum and statistical physics ({O}slo,
  1988)}, pages 370--389, 1992.

\bibitem{defant10}
A.~Defant, D.~Popa, and U.~Schwarting.
\newblock Coordinatewise multiple summing operators in {B}anach spaces.
\newblock {\em J. Funct. Anal.}, 259:220--242, 2010.

\bibitem{diniz12}
D.~Diniz, G.~Mu{\~n}oz-Fern{\'a}ndez, D.~Pellegrino, and
  J.~Seoane-Sep{\'u}lveda.
\newblock The asymptotic growth of the constants in the {B}ohnenblust-{H}ille
  inequality is optimal.
\newblock {\em J. Funct. Anal.}, 263:415--428, 2012.
\newblock \url{arXiv:1108.1550}.

\bibitem{dinur05}
I.~Dinur and E.~Friedgut.
\newblock Analytical methods in combinatorics and computer-science, 2005.
\newblock \url{http://www.cs.huji.ac.il/~analyt/}.

\bibitem{dinur07}
I.~Dinur, E.~Friedgut, G.~Kindler, and R.~O'Donnell.
\newblock On the {F}ourier tails of bounded functions over the discrete cube.
\newblock {\em Israel Journal of Mathematics}, 160:389--412, 2007.

\bibitem{ford05}
J.~Ford and A.~G{\'a}l.
\newblock Hadamard tensors and lower bounds on multiparty communication
  complexity.
\newblock In {\em Proc. 32\textsuperscript{nd} {I}nternational {C}onference on
  {A}utomata, {L}anguages and {P}rogramming (ICALP'05)}, pages 1163--1175,
  2005.

\bibitem{gallier09}
J.~Gallier.
\newblock Notes on spherical harmonics and linear representations of {L}ie
  groups, 2009.
\newblock \url{http://www.cis.upenn.edu/~cis610/sharmonics.pdf}.

\bibitem{gavinsky07}
D.~Gavinsky, J.~Kempe, I.~Kerenidis, R.~Raz, and R.~de~Wolf.
\newblock Exponential separations for one-way quantum communication complexity,
  with applications to cryptography.
\newblock In {\em Proc. 39\textsuperscript{th} Annual ACM Symp. Theory of
  Computing}, pages 516--525, 2007.
\newblock \url{quant-ph/0611209}.

\bibitem{gross75a}
L.~Gross.
\newblock Logarithmic {S}obolev inequalities.
\newblock {\em Amer. J. Math.}, 97(4):1061--1083, 1975.

\bibitem{gross06}
L.~Gross.
\newblock Hypercontractivity, logarithmic {S}obolev inequalities, and
  applications: a survey of surveys.
\newblock In {\em Diffusion, Quantum Theory, and Radically Elementary
  Mathematics}, pages 45--74. Princeton University Press, 2006.

\bibitem{harrow10}
A.~Harrow and A.~Montanaro.
\newblock An efficient test for product states, with applications to quantum
  {M}erlin-{A}rthur games.
\newblock In {\em Proc. 51\textsuperscript{st} Annual Symp. Foundations of
  Computer Science}, pages 633--642, 2010.
\newblock \url{arXiv:1001.0017}.

\bibitem{hartmann04}
M.~Hartmann, G.~Mahler, and O.~Hess.
\newblock Gaussian quantum fluctuations in interacting many particle systems.
\newblock {\em Lett. Math. Phys.}, 68(103--112), 2004.
\newblock \url{math-ph/0312045}.

\bibitem{helstrom76}
C.~W. Helstrom.
\newblock {\em Quantum detection and estimation theory}.
\newblock Academic Press, New York, 1976.

\bibitem{holevo73a}
A.~S. Holevo.
\newblock Statistical decision theory for quantum systems.
\newblock {\em Journal of Multivariate Analysis}, 3:337--394, 1973.

\bibitem{janson97}
S.~Janson.
\newblock {\em Gaussian {H}ilbert spaces}.
\newblock Cambridge University Press, 1997.

\bibitem{kahn88}
J.~Kahn, G.~Kalai, and N.~Linial.
\newblock The influence of variables on {B}oolean functions.
\newblock In {\em Proc. 29\textsuperscript{th} Annual Symp. Foundations of
  Computer Science}, pages 68--80, 1988.

\bibitem{kastoryano12}
M.~Kastoryano and K.~Temme.
\newblock Quantum logarithmic {S}obolev inequalities and rapid mixing, 2012.
\newblock {\tt arXiv:1207.3261}.

\bibitem{khot07}
S.~Khot, G.~Kindler, E.~Mossel, and R.~O'Donnell.
\newblock Optimal inapproximability results for {MAX-CUT} and other 2-variable
  {CSP}s?
\newblock {\em SIAM J. Comput.}, 37(1):319--357, 2007.

\bibitem{king12}
C.~King.
\newblock Hypercontractivity for semigroups of unital qubit channels, 2012.
\newblock \url{arXiv:1210.8412}.

\bibitem{klartag11}
B.~Klartag and O.~Regev.
\newblock Quantum one-way communication can be exponentially stronger than
  classical communication.
\newblock In {\em Proc. 43\textsuperscript{rd} Annual ACM Symp. Theory of
  Computing}, pages 31--40, 2011.
\newblock \url{arXiv:1009.3640}.

\bibitem{lancien12}
C.~Lancien and A.~Winter.
\newblock Distinguishing multi-partite states by local measurements, 2012.
\newblock \url{arXiv:1206.2884}.

\bibitem{littlewood30}
J.~Littlewood.
\newblock On bounded bilinear forms in an infinite number of variables.
\newblock {\em Quarterly Journal of Mathematics}, os-1:164--174, 1930.

\bibitem{matthews09a}
W.~Matthews, S.~Wehner, and A.~Winter.
\newblock Distinguishability of quantum states under restricted families of
  measurements with an application to quantum data hiding.
\newblock {\em Comm. Math. Phys.}, 291(3):813--843, 2009.
\newblock \url{arXiv:0810.2327}.

\bibitem{qboolean}
A.~Montanaro and T.~Osborne.
\newblock Quantum boolean functions.
\newblock {\em Chicago Journal of Theoretical Computer Science}, 2010.
\newblock \url{arXiv:0810.2435}.

\bibitem{mossel10}
E.~Mossel, R.~O'Donnell, and K.~Oleszkiewicz.
\newblock Noise stability of functions with low influences: Invariance and
  optimality.
\newblock {\em Ann. of Math.}, 171(1), 2010.

\bibitem{nelson66}
E.~Nelson.
\newblock A quartic interaction in two dimensions.
\newblock In {\em Mathematical Theory of Elementary Particles (Dedham,
  Massachusetts, 1965)}, pages 69--73. MIT Press, 1966.

\bibitem{nunezalarcon12}
D.~Nu{\~n}ez-Alarc\'on and D.~Pellegrino.
\newblock On the growth of the optimal constants of the multilinear
  {B}ohnenblust-{H}ille inequality, 2012.
\newblock \url{arXiv:1205.2385}.

\bibitem{odonnell07}
R.~O'Donnell.
\newblock 15-859{S}: {A}nalysis of boolean functions, 2007.
\newblock \url{http://www.cs.cmu.edu/~odonnell/boolean-analysis/}.

\bibitem{odonnell08}
R.~O'Donnell.
\newblock Some topics in analysis of boolean functions.
\newblock In {\em Proc. 40\textsuperscript{th} Annual ACM Symp. Theory of
  Computing}, pages 569--578, 2008.

\bibitem{pellegrino12}
D.~Pellegrino and J.~Seoane-Sep{\'u}lveda.
\newblock New upper bounds for the constants in the {B}ohnenblust-{H}ille
  inequality.
\newblock {\em Journal of Mathematical Analysis and Applications},
  386(1):300--307, 2012.
\newblock \url{arXiv:1010.0461v3}.

\bibitem{perezgarcia08}
D.~P{\'e}rez-Garc{\'i}a, M.~Wolf, C.~Palazuelos, I.~Villanueva, and M.~Junge.
\newblock Unbounded violation of tripartite {B}ell inequalities.
\newblock {\em Comm. Math. Phys.}, 279(2):455--486, 2008.
\newblock \url{quant-ph/0702189}.

\bibitem{roth08}
R.~Roth and K.~Viswanathan.
\newblock On the hardness of decoding the gale-berlekamp code.
\newblock {\em IEEE Trans. Inform. Theory}, 54(3):1050--1060, 2008.

\bibitem{sen06}
P.~Sen.
\newblock Random measurement bases, quantum state distinction and applications
  to the hidden subgroup problem.
\newblock In {\em Proc. 21\textsuperscript{st} Annual IEEE Conf. Computational
  Complexity}, page 287, 2006.
\newblock \url{quant-ph/0512085}.

\bibitem{serranorodriguez12}
D.~Serrano-Rodr\'iguez.
\newblock A closed form for subexponential constants in the
  {B}ohnenblust-{H}ille inequality, 2012.
\newblock \url{arXiv:1205.4735}.

\bibitem{simon72}
B.~Simon and R.~H{\o }egh-Krohn.
\newblock Hypercontractive semigroups and two dimensional self-coupled {B}ose
  fields.
\newblock {\em J. Funct.\ Anal.}, 9(2):121--180, 1972.

\bibitem{stein71}
E.~Stein and G.~Weiss.
\newblock {\em Fourier analysis on {E}uclidean spaces}.
\newblock Princeton University Press, 1971.

\bibitem{dewolf08}
R.~de Wolf.
\newblock A brief introduction to {F}ourier analysis on the boolean cube.
\newblock {\em Theory of Computing Library Graduate Surveys}, 1:1--20, 2008.

\end{thebibliography}

\end{document}